\documentclass[final,11pt,3p]{elsarticle}

\usepackage{graphicx,amsfonts,amsmath,amssymb,amsthm,color,natbib} 

\newcommand*{\R}{\mathbb{R}}

\newcommand*{\C}{\mathbb{C}}

\theoremstyle{plain}
\newtheorem{theo}{Theorem}[section]

\newtheorem{propo}[theo]{Proposition}

\newtheorem{remark}{Remark}

\theoremstyle{definition}

\def\ge{g_\text e} 
\def\gen{g_{{\rm e},n}}
\def\ve{v_\text e}
\def\d{\text d}
\def\bb{\langle b\rangle}
\def\bia{\langle 1/a\rangle}

\def\fgn{f_n^g}
\def\fvn{f_n^v}
\def\fdvn{f_n^{\Delta v}}
\def\fg{f_g}
\def\fv{f_v}
\def\fdv{f_{\Delta v}}

\def\tmin{T_{\rm min}}
\def\tmax{T_{\rm max}}

\usepackage{xcolor}
 
\begin{document}
\begin{frontmatter}

\title{Stability of heterogeneous linear and nonlinear\\ car-following models}


\author[ACM,IMACM]{Matthias Ehrhardt\corref{Corr}}
\cortext[Corr]{Corresponding author}
\ead{ehrhardt@uni-wuppertal.de}



\author[TSR,IMACM]{Antoine Tordeux}
\ead{tordeux@uni-wuppertal.de}


\address[ACM]{University of Wuppertal, Chair of Applied and Computational Mathematics,\\
Gaußstrasse 20, 42119 Wuppertal, Germany}

\address[TSR]{University of Wuppertal, Chair of Traffic Safety and Reliability,\\
Gaußstrasse 20, 42119 Wuppertal, Germany}

\address[IMACM]{Institute for Mathematical Modelling, Analysis and Computational Mathematics,\\
Gaußstrasse 20, 42119 Wuppertal, Germany}

\begin{abstract}
 Stop-and-go waves in road traffic are complex collective phenomena with significant implications for traffic engineering, safety and the environment. 
 Despite decades of research, understanding and controlling these dynamics remains challenging. 
 This article examines two classes of heterogeneous car-following models with quenched disorder to shed light on the underlying mechanisms that drive traffic instability and stop-and-go dynamics.
 Specifically, a scaled heterogeneity model and an additive heterogeneity model are investigated, each of which affects the stability of linear and nonlinear car-following models differently. 
 We derive general linear stability conditions which we apply to specific models and illustrate by simulation.
 The study provides insights into the role of individual heterogeneity in vehicle behaviour and its influence on traffic stability.
\end{abstract}

\begin{keyword}
car-following model \sep quenched disorder \sep collective stability \sep stop-and-go waves.\\
\textit{2020 Mathematics Subject Classification:} 76A30, 82C22.
\end{keyword}

\journal{Franklin Open}


\end{frontmatter}

\section{Introduction}

Stop-and-go waves in road traffic are fascinating collective phenomena. 
Stop-and-go waves and traffic instability are observed daily around the world. They have also been observed in experiments \cite{sugiyama2008traffic, nakayama2009metastability, tadaki2013phase}, both with human drivers and with automated driver assistance systems (adaptive cruise control systems) \cite{stern_DissipationStopandgoWaves_2018, gunter2020commercially, makridis2021openacc}.
Even deep learning approaches are being developed to dissipate stop-and-go waves \cite{kreidieh2018dissipating, jiang2021dampen}.
In addition to the scientific interest in traffic engineering, stop-and-go traffic dynamics pose important challenges for road safety and the environment. 
Indeed, stop-and-go dynamics lead to excessive fuel consumption and pollutant emissions compared to a uniform flow of traffic \cite{andre2000driving,aguilera2014new, stern_DissipationStopandgoWaves_2018, stern2019quantifying}.

Stop-and-go waves and traffic instability are classically addressed in traffic engineering using microscopic car-following models. 
Pioneering work in these areas dates back to the studies of Reuschel and Pipes early 1950s \cite{Reuschel1950fahrzeugbewegungen, pipes_OperationalAnalysisTraffic_1953}, Chandler, Herman and co-authors late 1950s with linear models \cite{kometani1958stability, chandler1958traffic, herman_TrafficDynamicsAnalysis_1959}, and, later, Bando, Jiang, Treiber and co-authors with non-linear models \cite{bando_DynamicalModelTraffic_1995, bando1998analysis, jiang_FullVelocityDifference_2001, treiber_CongestedTrafficStates_2000, tordeux2010adaptive}.
We refer the interested reader to \cite{wilson_CarfollowingModelsFifty_2011, cordes_ModelingStopandGoWaves_2020, cordes2023single} for a review.

Three main factors have been identified as potential causes of stop-and-go waves:
\begin{enumerate}
    \item {\bf Delay and relaxation}. 
    In fact, the introduction of a delay or relaxation in the dynamics of most car-following models leads to stability breaking. Such a feature has been noted since the 1950-1960s with linear and nonlinear models \cite{kometani1958stability, chandler1958traffic, gazis1961nonlinear, newell1961nonlinear} and is now a consensus in traffic engineering \cite{nagatani1998delay, orosz2004global, orosz2010traffic, tordeux_LinearStabilityAnalysis_2012, tordeux2018traffic}.
    \item {\bf Stochastic noise}. 
    The first studies in this field were carried out using discrete cellular automata \cite{BarlovicSSS98, Ke-Ping_2004,kaupuvzs2005zero, huang2018instability, schadschneider_StochasticTransportComplex_2010}. 
    More recently, continuous approaches are based on Langevin and Ornstein-Uhlenbeck processes \cite{tomer2000presence, wagner2011time, treiber2009hamilton, treiber2017intelligent, Tordeux2016b, ngoduy2019langevin, wang2020stability, friesen2021spontaneous, ackermann_tordeux24, Ehrhardt2024}. 
    \item {\bf Individual heterogeneity}. 
    Recent approaches show that heterogeneity in vehicle behaviour can induce stability breaking or, conversely, improve the stability of vehicle trains \cite{stern_DissipationStopandgoWaves_2018}. 
    Heterogeneity models can be distinguished between static models lying in the agent characteristics (quenched disorder), and dynamic heterogeneity models acting in the interactions (annealed disorder) \cite{krusemann2014first,tateishi2020quenched, khelfa2022heterogeneity}.
    A general linear stability condition for heterogeneous car-following models with quenched disorder is given in \cite{ngoduy2015effect}.
\end{enumerate}
Despite the large number of studies, understanding and controlling stop-and-go in road traffic flow remains challenging and still nowadays an active area of research.
In particular, the role of heterogeneity and non-linearity in the shape of the model remains poorly understood.


In this article, we investigate two classes of heterogeneous car-following models with quenched disorder: an scaled heterogeneity model where a vehicle-specific factor affects the car-following model and an additive heterogeneity model, where, similar to a noisy model, an individual bias is added to the model. 
If the scaled heterogeneity model can affect the stability of both linear and non-linear car-following models, 
additive heterogeneity terms can disturb the stability of nonlinear models; this is not the case for linear models. 
The \textit{full velocity difference} (FVD) model \cite{jiang_FullVelocityDifference_2001} and the \textit{adaptive time gap} (ATG) model are used as reference linear and nonlinear models, respectively.

The paper is structured as follows.
Next, we present the system setup, notations and a summary of the main results. 
The car-following models, including homogeneous, heterogeneous, linear and nonlinear models, are defined in section~\ref{sec:DefMod}. 
The linear stability conditions are derived in section~\ref{sec:LinStab}.
Some simulations illustrate the results in section~\ref{sec:Sim}. 

\subsection{Setup and notations}
In the sequel, we will consider $N=3,4,\dots$ vehicles of \textit{length} $\ell>0$  on a segment of length $L>N\ell$ with periodic boundaries, i.e.\ a circular road.
We denote by $(x_n(t))_{n=1}^N$ the \textit{positions} of the vehicles at time $t\geq0$ and assume that the vehicles are initially ordered by their index $n$, i.e.,
\begin{equation*}
   0\le x_1(0)\le x_2(0)\le \dots \le x_N(0) \le L.
\end{equation*}
We assume that the predecessor of the $n$th vehicle is always the ($n+1$)th vehicle, 
while the predecessor of the $N$th vehicle is the first vehicle due to the periodic boundaries. 
The \textit{distance gaps} to the predecessors are the variables
\begin{equation}
\begin{cases}
    &g_n(t)= x_{n+1}(t)-x_n(t)-\ell,\qquad n\in\{1,\dots,N-1\},\\
    &g_N(t)= L+x_{1}(t)-x_N(t)-\ell.
\end{cases}
\end{equation}
The \textit{speeds} of the vehicles at time $t\geq0$ are denoted by $(v_n(t))_{n=1}^N$, 
where $v_n(t)=\d x_n(t)/\d t = \dot{x}_n(t)$ for all $n\in\{1,\dots,N\}$. 
The \textit{speed differences} with the predecessors are given by
\begin{equation}
\begin{cases}
    &\Delta v_n(t)=v_n(t)-v_{n+1}(t),\qquad n\in\{1,\dots,N-1\},\\
    &\Delta v_N(t)=v_N(t)-v_{1}(t).
\end{cases}
\end{equation}

\subsection{Summary of the results}

We consider a homogeneous car-following model, defined by some universal acceleration function $F$,
and the equilibrium 
gap $\ge$ and the unique equilibrium speed $\ve$: 
\begin{equation}\label{eq:CFsumm}
    \dot{v}_n(t) = F\bigl(g_n(t),v_n(t),\Delta v_n(t)\bigr), \qquad \ge=:L/N-\ell>0,
    \qquad\exists!\, \ve\in\R,\;F(\ge,\ve,0)=0,
\end{equation}
 and analyze the linear stability of two types of heterogeneous dynamics, namely
\begin{enumerate}
    \item {\bf Scaled heterogeneity model}
\begin{equation}\label{eq:HCF1summ}
    \dot{v}_n(t) = a_n F\bigl(g_n(t),v_n(t),\Delta v_n(t)\bigr),\qquad a_n>0.
\end{equation}
    \item {\bf Additive heterogeneity model}
\begin{equation}\label{eq:HCF2summ}
    \dot{v}_n(t) = F\bigl(g_n(t),v_n(t),\Delta v_n(t)\bigr) + b_n,\qquad b_n\in\R.
\end{equation}
\end{enumerate}

The results show that the scaled heterogeneity factors $(a_n)_{n=1}^N$ in \eqref{eq:HCF1summ} can affect the linear stability of both linear and nonlinear car-following models $F$. 
However, the additive heterogeneity terms $(b_n)_{n=1}^N$ in \eqref{eq:HCF2summ} can only affect the stability of nonlinear models (see Table~\ref{tab:sum}).

\begin{table}[!ht]
    \centering\medskip
    \renewcommand{\arraystretch}{1.5}
    \begin{tabular}{p{2.7cm}|p{6cm}|p{6.5cm}}
        &\multicolumn{1}{c|}{\bf Scaled heterogeneity model}&\multicolumn{1}{c}{\bf Additive heterogeneity model\vspace{-10mm}}\\
        &$$ \begin{array}{c}
             \dot{v}_n(t)=a_n F(g_n(t),v_n(t),\Delta v_n(t))\\[-3mm]
             a_n>0
        \end{array}$$\vspace{-4mm} &  
        $$ 
        \begin{array}{c}
        \dot{v}_n(t) = F(g_n(t),v_n(t),\Delta v_n(t)) + b_n\\[-3mm]
        b_n\in\R
        \end{array}$$\vspace{-4mm}\\
        \hline
         {\bf Linear stability condition}&
        $$
         \frac{1}{2}\fv^2+\fv\fdv\geq\fg\bia
        $$ &
         $\bullet$~~\textit{Linear model:} ~{\bf as for the homogeneous model}: 
         $$\frac{1}{2}\fv^2+\fv\fdv\geq\fg$$\\
         &\vspace{-16mm}
         {\bf ~FVD model:} \vspace{-1mm}
         $$\frac{1}{2} T\lambda_1+T\lambda_2\geq \bia$$
         &\vspace{-7mm}
         $\bullet$~~\textit{Nonlinear model:} ~{\bf Model specific}\\
         &\vspace{-10mm}
         {\bf ~ATG model:} \vspace{-1mm}
         $$\frac{1}{2} T\lambda+1\geq \bia$$\vspace{-5mm}
         &\vspace{-10mm}
         {\bf ~ATG model:} \vspace{-1mm}
         $$\sum_{n=1}^N\;\frac{b_n(2\lambda T+1)}{(b_n+\lambda \ve)^3}
+\frac{\lambda^3T\ve^2}{2(b_n+\lambda \ve)^4}
    \geq0.$$\vspace{-4mm}
    \end{tabular}\medskip
    \caption{Summary of the results. Notations: $\fg=\partial_gF(\ge,\ve,0)$, $\fv=\partial_vF(\ge,\ve,0)$, $\fdv=\partial_{\Delta v}F(\ge,\ve,0)$, 
		and $\bia=\frac{1}{N}\sum_{n=1}^N1/a_n$. 
  The linear \textit{full velocity difference} (FVD) model \cite{jiang_FullVelocityDifference_2001} is given by $F(g,v,\Delta v)=\lambda_1(g/T-v)-\lambda_2\Delta v$ with $\lambda_1,\lambda_2,T>0$, while the \textit{adaptive time gap} (ATG) nonlinear model \cite{tordeux2010adaptive} reads $
    F(g,v,\Delta v)=\lambda v (1-Tv/g)-v\Delta v/g$ with $\lambda,T>0$.}
    \label{tab:sum}
\end{table}

\section{Car-following models and equilibrium solutions}\label{sec:DefMod}

\subsection{Homogeneous car-following models}

A general homogeneous car-following model reads
\begin{equation}\label{eq:CF}
\begin{cases}
    ~\dot{x}_n(t)= v_n(t)\\
    ~\dot{v}_n(t) = F\bigl(g_n(t),v_n(t),\Delta v_n(t)\bigr).
\end{cases}
\end{equation}
We assume that the universal acceleration function $F\colon\R^3\mapsto\R$ is monotone non-decreasing with its first argument and that there exists for any $\ge=L/N-\ell>0$ a unique equilibrium speed $\ve$ such that
\begin{equation}
    F(\ge,\ve,0)=0.
\end{equation}
Both equilibrium gap $\ge$ and equilibrium speed $\ve$ are parameterised by the vehicle concentration on the segment.

\paragraph{FVD and ATG models}
The linear \textit{full velocity difference} (FVD) model \cite{jiang_FullVelocityDifference_2001} is given by
\begin{equation}\label{eq:FVD}
    F(g,v,\Delta v) = \lambda_1\Bigl(\frac{g}{T}-v\Bigr) - \lambda_2\Delta v,\qquad \lambda_1,\lambda_2,T>0,
\end{equation}
and the nonlinear \textit{adaptive time gap} (ATG) model \cite{tordeux2010adaptive} is given by
\begin{equation}\label{eq:ATG}
    F(g,v,\Delta v)=\lambda v \Bigl(1-\frac{Tv}{g}\Bigr) - \frac{v\Delta v}{g},\qquad \lambda,T>0.
\end{equation}
For these two models, the equilibrium speed $\ve$ is 
\begin{equation}\label{eq:EVHomo}
   \ve = \frac{\ge}{T}=\frac{L/N-\ell}{T}.
\end{equation}

\subsection{Heterogeneous car-following models}
A general heterogeneous car-following model with quenched disorder reads
\begin{equation}\label{eq:HCF}
    \begin{cases}
    ~\dot{x}_n(t)= v_n(t)\\
    ~\dot{v}_n(t) = F_n\bigl(g_n(t),v_n(t),\Delta v_n(t)\bigr),
    \end{cases}
\end{equation}
where the acceleration function $F_n\colon\R^3\mapsto\R$ is specific to each vehicle. 
We assume that there exists an equilibrium solution $\bigl((\gen)_{n=1}^N,\ve\bigr)$ such that
\begin{equation}
\begin{cases}
     ~\displaystyle\sum_{n=1}^N \gen = L-N\ell = N\ge,\\
     ~F_n(\gen,\ve,0)=0,\qquad  n\in\{1,\dots,N\}.
\end{cases}
\end{equation}
In the sequel, we will consider two families of heterogeneous car-following models with quenched disorder: scaled and additive heterogeneity models.

\subsubsection{Scaled heterogeneity model}

The scaled heterogeneity model reads
	\begin{equation}\label{eq:HCF1}
        \dot{v}_n(t) = a_n F\bigl(g_n(t),v_n(t),\Delta v_n(t)\bigr),\qquad a_n>0.
  \end{equation}
Here, $a_nF=0$ implies $F=0$ as $a_n>0$. 
Therefore, the equilibrium solution of the heterogeneous model \eqref{eq:HCF1} matches the solution of the homogeneous model, i.e., $\gen=\ge$ for all $n\in\{1,\dots,N\}$ while $\ve$ is the solution of the equilibrium equation $F(\ge,\ve,0)=0$.

\subsubsection{Additive heterogeneity model}

The additive heterogeneity model is given by
\begin{equation}\label{eq:HCF2}
    \dot{v}_n(t) = F\bigl(g_n(t),v_n(t),\Delta v_n(t)\bigr) + b_n,\qquad b_n\in\R.
\end{equation}
Here, the equilibrium solution satisfies the gap conservation and the equilibrium condition, i.e.
\begin{equation}\label{eq:VE}
\begin{cases}
     ~\displaystyle\sum_{n=1}^N \gen=N\ge\\
     ~F(\gen,\ve,0)+b_n=0,
     \qquad n\in\{1,\dots,N\}.
\end{cases}
\end{equation}

\paragraph{Full velocity difference model}
\begin{propo}
The equilibrium state \eqref{eq:VE} of the FVD model \eqref{eq:FVD} with additive bias is given by
\begin{equation}\label{eq:veFVD}
   \ve = \frac{\ge}{T} + \frac{\bb}{\lambda_1},
   \quad\text{with}\quad\bb=\frac{1}{N}\sum_{n=1}^Nb_n,
\end{equation}
and
\begin{equation}\label{eq:genFVD}
     \gen = \ge + \frac{T}{\lambda_1} \bigr(\bb-b_n\bigr), \qquad n\in\{1,\dots,N\}.
\end{equation}
\end{propo}

\begin{proof}
We obtain from \eqref{eq:FVD} and \eqref{eq:VE} that $\lambda_1 (\gen/T-\ve)+b_n=0$ for all $n\in\{1,\dots,N\}$ and we can deduce
\begin{equation}\label{eq19}
    \gen=T\Bigl(\ve-\frac{b_n}{\lambda_1}\Bigr),\qquad n\in\{1,\dots,N\}.
\end{equation}
Using \eqref{eq19} and the gap conservation in \eqref{eq:VE}, 
the equilibrium speed $\ve$ is the solution of 
\begin{equation}
\sum_{n=1}^N T \Big(\ve-\frac{b_n}{\lambda_1}\Big)=N\ge,
\end{equation}
which is
\begin{equation}
     \ve=\frac{\ge}{T}+\frac{\bb}{\lambda_1},\qquad \bb=\frac{1}{N}\sum_{n=1}^Nb_n.
\end{equation}
It follows from \eqref{eq19} that 
\begin{equation}
    \gen=T\Bigl(\frac{\ge}{T} +\frac{\bb}{\lambda_1}-\frac{b_n}{\lambda_1}\Bigr)
    =\ge + \frac{T}{\lambda_1}\bigr(\bb-b_n\bigr),\qquad n\in\{1,\dots,N\}.
\end{equation}
\end{proof}

The equilibrium speed \eqref{eq:veFVD} is non-negative if 
\begin{equation}\label{eq:c1FVD}
    \bb>-\frac{\ge\lambda_1}{T},
\end{equation}
while the equilibrium gaps \eqref{eq:genFVD} are non-negative for each vehicle if 
\begin{equation}\label{eq:c2FVD}
   \max_n b_n<\frac{\ge\lambda_1}{T}+\bb.
\end{equation}
Furthermore, the speed \eqref{eq:veFVD} matches the equilibrium speed of the homogeneous model
    $\ve=\ge/T$,
and the equilibrium gaps \eqref{eq:genFVD} are 
\begin{equation}
    \gen=\ge-b_n \frac{T}{\lambda_1},\qquad n\in\{1,\dots,N\},
    \end{equation}
if the individual bias $b_n$ is zero in average, i.e., if $\bb=0$.

\begin{remark}[Homogeneous case]
If the individual bias $b_n$ is identical for all vehicles,   $b_n=b$, $n\in\{1,\dots,N\}$,
then the equilibrium speed \eqref{eq:veFVD} is
\begin{equation}
    \ve=\frac{\ge}{T} + \frac{b}{\lambda_1},
\end{equation}
while the equilibrium gaps \eqref{eq:genFVD} are uniform: $\gen=\ge$, $n\in\{1,\dots,N\}$.
\end{remark}

\paragraph{Adaptive time gap model}

\begin{propo}
The equilibrium state \eqref{eq:VE} of the ATG model \eqref{eq:ATG} with additive bias is given by
\begin{equation}\label{eq:veATG}
    \sum_{n=1}^N \frac{\lambda T\ve^2}{b_n+\lambda \ve}=N\ge
\end{equation}
and
\begin{equation}\label{eq:genATG}
    \gen=\frac{\lambda T\ve^2}{b_n+\lambda \ve},\quad n\in\{1,\dots,N\}.
\end{equation}
\end{propo}
\begin{proof}
We obtain from \eqref{eq:ATG} and \eqref{eq:VE} that $\lambda\ve(1-T\ve/\gen)+b_n=0$ and we can deduce
\begin{equation}
    \gen=\frac{T\ve}{\frac{b_n}{\lambda \ve}+1}
    =\frac{\lambda T\ve^2}{b_n+\lambda \ve},\qquad n\in\{1,\dots,N\}.
\end{equation}
Using the gap conservation $\sum_{n=1}^N \gen=N\ge$, the equilibrium speed $\ve$ is the solution of \eqref{eq:veATG}.
\end{proof}

Let us note that the equilibrium gaps $\gen$ \eqref{eq:genATG} are positive for each vehicle if
\begin{equation}\label{eq:c1ATG}
    \ve> -\frac{1}{\lambda}\min_n b_n.
\end{equation}

\begin{remark}
We recover the equilibrium solution of the homogeneous ATG model \eqref{eq:EVHomo}
\begin{equation}
\ve=\frac{\ge}{T}\quad\text{and}\quad
\gen=\ge,\quad n\in\{1,\dots,N\},
\end{equation} 
if the biases are zero, i.e., $b_n=0$ for all $n\in\{1,\dots,N\}$.
\end{remark}

In the case where the individual bias is the same 
for all vehicles: $b_n=b$, $n\in\{1,\dots,N\}$, we have 
\begin{equation}\label{eq:veATG1}
    \frac{\lambda T\ve^2}{b+\lambda \ve}=\ge.
\end{equation}
Assuming $b>-\lambda\ve$ (see \eqref{eq:c1ATG}), we obtain
\begin{equation}
    \lambda T\ve^2-(b+\lambda \ve)\ge=0,
\end{equation}
and we can deduce that (we take the positive root to make $\ve$ positive) 
\begin{equation}\label{eq:veATG2}
    \ve = \frac{\ge\lambda+\sqrt{(\ge\lambda)^2+4\lambda Tb\ge}}{2\lambda T}
    = \frac{\ge}{2T} \Bigl(1+\sqrt{1+\frac{4Tb}{\lambda\ge}}\Bigr).
\end{equation}
The equilibrium speed $\ve$ exists if $1+\frac{4Tb}{\lambda\ge}\geq0$ and we obtain the condition
\begin{equation}\label{eq:c2ATG}
    b\geq -\frac{\lambda\ge}{4T}.
\end{equation}
Note that \eqref{eq:c2ATG} implies the preliminary condition $b>-\lambda\ve$.

\section{Linear stability analysis} \label{sec:LinStab}

\subsection{General linear stability condition}

The partial derivatives of the general heterogeneous car-following model \eqref{eq:HCF} evaluated at the equilibrium are given by
\begin{equation}
    \fgn=\frac{\partial F_n}{\partial g}(\gen,\ve,0),
    \qquad
    \fvn=\frac{\partial F_n}{\partial v}(\gen,\ve,0),
    \qquad
    \fdvn=\frac{\partial F_n}{\partial \Delta v}(\gen,\ve,0).
\end{equation}
The characteristic equation of the resulting linear ODE system reads
\begin{equation}\label{eq:EC}
    \prod_{n=1}^N \bigl[z^2-z(\fvn-\fdvn)+\fgn\bigr] -e^{-{i\theta N}} \prod_{n=1}^N \bigl[-z\fdvn+\fgn\bigr]=0,
    \qquad z\in\C,\;\theta\in[0,2\pi].
\end{equation}
A sufficient general linear stability condition for which all eigenvalues $z$ have non-positive real parts, except one equal to zero (due to the periodic boundaries), 
is given by \cite[Eq.~(5)]{ngoduy2015effect} 
\begin{equation}\label{eq:SC}
    \sum_{n=1}^N \biggl[\frac{1}{2}\Bigl(\frac{\fvn}{\fgn}\Bigr)^2 +\frac{\fvn\fdvn}{\fgn\fgn}-\frac{1}{\fgn}\biggr]\geq0.
\end{equation}

For homogeneous models with
$\fgn=\fg$, $\fvn=\fv$, and $\fdvn=\fdv$, for all $n\in\{1,\dots,N\}$,
assuming $\fg>0$, the linear stability condition \eqref{eq:SC} is given by \cite{tordeux_LinearStabilityAnalysis_2012,treiber2013traffic}
\begin{equation}\label{eq:SCHomo}
    \frac{1}{2}\fv^2+\fv\fdv\geq \fg.
\end{equation}
For instance, we have for the FVD model \eqref{eq:FVD} 
\begin{equation}
     \fg=\frac{\lambda_1}{T}>0,\quad\fv = -\lambda_1, \quad\text{and}\quad \fdv = -\lambda_2,
\end{equation} 
and thus the linear stability condition reads \cite{jiang_FullVelocityDifference_2001}
\begin{equation}\label{eq:SC-FVD}
    \frac{1}{2}\lambda_1 + \lambda_2 \geq \frac{1}{T}.
\end{equation}
In case of the ATG model \eqref{eq:ATG} we have 
\begin{equation}\label{eq:SC-ATG}
    \fg=\frac{\lambda}{T}>0, \quad
    \fv=-\lambda, \quad\text{and}\quad \fdv=-\frac{1}{T},
\end{equation}
and the linear stability condition
\begin{equation}
    \frac{\lambda T}2+1 \geq 1
\end{equation}
systematically holds since $\lambda,T>0$. 
Indeed, the ATG model is unconditionally linearly stable, cf.\  \cite{khound2021extending}.

\subsection{Scaled heterogeneity model}

\begin{propo}
Assuming $\fg>0$, 
\begin{equation}\label{eq:SC1}
    \frac{1}{2}\fv^2+\fv\fdv\geq \bia\fg,\quad\text{with}\quad \bia=\frac{1}{N}\sum_{n=1}^N \frac{1}{a_n},
\end{equation}
is a sufficient linear stability condition of the scaled heterogeneity model \eqref{eq:HCF1}.
\end{propo}

\begin{proof}
We have for the scaled heterogeneity model \eqref{eq:HCF1} 
\begin{equation}
    \fgn=a_n\fg,\qquad\fvn=a_n\fv,\qquad\fdvn=a_n\fdv.
\end{equation}
Using the general linear stability condition \eqref{eq:SC} with these partial derivatives allows recovering directly the condition \eqref{eq:SC1}.
\end{proof}

\begin{remark} 
The condition \eqref{eq:SC1} matches the linear stability condition \eqref{eq:SCHomo} of homogeneous models if $\bia=1$. 
\end{remark}

The condition for the FVD model \eqref{eq:FVD} reads
\begin{equation}
    \frac{1}{2}T\lambda_1+T\lambda_2\geq\bia,
\end{equation}
while the linear stability condition for the ATG model \eqref{eq:ATG} is given by 
\begin{equation}
     \frac{\lambda T}2+1 \geq \bia.
\end{equation}

Values of $a_n$ less than 1 affect negatively the stability and inversely. 
Due to the convexity of the inverse function, the weight of a term close to zero can be much higher than the weight of a high term. 
Assuming that $(a_n)_{n=1}^N$ are independent and randomly distributed according to a continuous distribution with density $h$ on $(0,\infty)$, it is easy to check that the inverse $1/a_n$ has a density $u^2h(1/u)$ on $(0,\infty)$. 
For instance, if $a_n$ is uniformly distributed in $[1-\kappa,1+\kappa]$ with $0<\kappa<1$, we obtain $\bia\to1+\kappa^2$. The scaled heterogeneity negatively affects the stability although it is symmetrically distributed around 1. 
Note that for this special case, the stability condition for the ATG model is asymptotically $\lambda T\geq2\kappa$. 

\subsection{Additive heterogeneity model}

\subsubsection{Linear models}

\begin{propo}\label{prop:SC_HCFlin}
For linear models, the stability condition of the additive heterogeneity model \eqref{eq:HCF2} matches the condition \eqref{eq:SCHomo} of homogeneous models. 
In other words, the additive heterogeneity does not affect the stability of linear models.
\end{propo}

\begin{proof}
The partial derivatives are constant for linear models even in presence of additive biases, i.e.,
\begin{equation}
    \fgn=\fg,\qquad\fvn=\fv,\qquad\fdvn=\fdv.
\end{equation}
for all $n\in\{1,\dots,N\}$.
Therefore, the stability condition for linear models is the same as the 
condition for homogeneous models given in \eqref{eq:SCHomo}.
\end{proof}

\subsubsection{Nonlinear models}

By definition, the partial derivatives at equilibrium of nonlinear models depend on the equilibrium state. 
The equilibrium state being specific to the vehicle for the additive heterogeneity model, see \eqref{eq:genFVD} of \eqref{eq:genATG}, it turns out that the partial derivatives are heterogeneous. 
The dependencies with the biases depends on the nature of the nonlinear components of the model.
This makes the linear stability condition model specific.

\begin{propo}
A sufficient linear stability condition of the ATG model \eqref{eq:ATG} with additive heterogeneity \eqref{eq:HCF2} is given by
\begin{equation}\label{eq:SCatg2}
   \sum_{n=1}^N\;\frac{b_n(2\lambda T+1)}{(b_n+\lambda \ve)^3}
   +\frac{\lambda^3T\ve^2}{2(b_n+\lambda \ve)^4}
    \geq0.
\end{equation}
\end{propo}

\begin{proof}
For the ATG model \eqref{eq:ATG}, we obtain the partial derivatives
\begin{equation}
    \fgn=
    \frac{\lambda T\ve^2}{(\gen)^2},
    \qquad
    \fvn=
    \lambda\Bigl(1-\frac{2T\ve}{\gen}\Bigr),
    \qquad
    \fdvn=
    -\frac{\ve}{\gen}\\
\end{equation}
We have
\begin{equation}
    \frac{\fvn}{\fgn}
    =\frac{\lambda\Bigl(1-\frac{2T\ve}{\gen}\Bigr)}{\frac{\lambda T\ve^2}{(\gen)^2}}
    =\frac{\gen(\gen-2T\ve)}{T\ve^2},
\end{equation}
while
\begin{equation}
    \frac{\fvn\fdvn}{\fgn\fgn}
    =-\frac{\lambda\Bigl(1-\frac{2T\ve}{\gen}\Bigr)\frac{\ve}{\gen}}{\frac{\lambda^2 T^2\ve^4}{(\gen)^4}}
    =-\frac{(\gen)^2(\gen-2T\ve)}{\lambda T^2\ve^3}.
\end{equation}
The sufficient linear stability condition \eqref{eq:SC} is then given by
\begin{equation}
    \sum_{n=1}^N \biggl[\frac{1}{2}\Bigl(\frac{\gen(\gen-2T\ve)}{T\ve^2}\Bigr)^2
    -\frac{(\gen)^2(\gen-2T\ve)}{\lambda T^2\ve^3}
    -\frac{(\gen)^2}{\lambda T\ve^2}\biggr]\geq0,
\end{equation}
or again
\begin{multline}
    \sum_{n=1}^N (\gen)^2 \biggl[\frac{\lambda(\gen-2T\ve)^2}{2T\ve^2}
    -\frac{\gen-2T\ve}{T\ve}-1\biggr]
    \\= \sum_{n=1}^N\;(\gen)^2 \biggl[\frac{\lambda(\gen-2T\ve)^2}{2T\ve^2}
    -\frac{\gen-T\ve}{T\ve}\biggr]\geq0.
\end{multline}
Then, using
$\gen=\frac{\lambda T\ve^2}{b_n+\lambda \ve}$ 
and remarking that
$\gen-2T\ve=-T\ve\frac{2b_n+\lambda \ve}{b_n+\lambda \ve}$
while
$\gen-T\ve=-T\ve\frac{b_n}{b_n+\lambda \ve}$,
we obtain 
\begin{equation}
    \sum_{n=1}^N \Bigl(\frac{\lambda T\ve^2}{b_n+\lambda \ve}\Bigr)^2 
    \Biggl[\frac{\Bigl(T\ve\frac{2b_n+\lambda \ve}{b_n+\lambda \ve}\Bigr)^2}{2T\ve^2}
    +\frac{T\ve\frac{b_n}{b_n+\lambda \ve}}{T\ve}\Biggr]
    \geq0.
\end{equation}
After simplifications (we have $\lambda,T,\ve>0$), it follows
\begin{equation}
    \sum_{n=1}^N \frac{1}{(b_n+\lambda \ve)^2}
    \biggl[\frac{\lambda T}2\Bigl(\frac{2b_n+\lambda \ve}{b_n+\lambda \ve}\Bigr)^2
    +\frac{b_n}{b_n+\lambda \ve}\biggr]
    \geq0,
\end{equation}
or again
\begin{equation}\sum_{n=1}^N \frac{1}{(b_n+\lambda \ve)^4}
    \biggl[\frac{\lambda T}{2} (2b_n+\lambda \ve)^2
    +b_n(b_n+\lambda \ve)\biggr]
    \geq0.
\end{equation}
We have
\begin{equation}\label{eq:intermed}
\begin{split}
    \frac{\lambda T}{2} (2b_n+\lambda \ve)^2 + b_n(b_n+\lambda \ve)
    &=2\lambda Tb_n^2 +\frac{1}{2}\lambda^3T\ve^2 +2\lambda^2Tb_n\ve+b_n^2+b_n\lambda\ve\\
    &=b_n^2(2\lambda T+1)+b_n(2\lambda T+1)\lambda\ve + \frac{1}{2}\lambda^3T\ve^2\\
    &=b_n(2\lambda T+1)(b_n+\lambda\ve)+\frac{1}{2}\lambda^3T\ve^2,
\end{split}
\end{equation}
and the linear stability condition can be written as \eqref{eq:SCatg2}.
\end{proof}


\begin{remark}
    Note that $\lambda,T>0$ and the model is unconditionally linearly stable if $b_n=0$ for all $n\in\{1,\dots,N\}$. 
    Indeed, the homogeneous ATG model is unconditionally linearly stable \cite{khound2021extending}. 
\end{remark} 

    The stability condition \eqref{eq:SCatg2} holds systematically if the biases $b_n\geq0$, $n\in\{1,\dots,N\}$ are all positive. 
    Since traffic performance is improved when the biases are positive, this gives us a way to improve the stability of platooning systems. 
    Conversely, negative biases are needed to destabilise the system. 

    More precisely, the right-hand term in \eqref{eq:SCatg2} is always positive, but the left-hand term can be negative if $b_n<0$ and even strongly negative (potentially breaking stability) if some 
    $b_n\xrightarrow[+]{}-\lambda\ve$. 
    However, this term is offset by the positive term, which also explodes when some $b_n\xrightarrow[+]{}-\lambda\ve$,
    and even faster than the negative term (exponent 4 versus exponent 3). 
    This suggests that moderately negative $b_n\in[-\lambda\ve,0]$ can destabilise the system.

    When all individual biases are the same, 
    i.e., $b_n=b>-\lambda\ge/(4T)$ for all $n\in\{1,\dots,N\}$, dividing by $\lambda T\ve^2$ the last line of \eqref{eq:intermed} and using \eqref{eq:veATG1}, we obtain the linear stability condition of the ATG model 
\begin{equation}
    b\geq\frac{-\lambda^2\ge}{4T\lambda+2}.
\end{equation}
The bias has to be negative and sufficiently low, especially for high $\lambda$ or low $T$, to destabilise the system.

\section{Numerical results} \label{sec:Sim}
In this section we present some numerical results to illustrate the investigations. 
We focus on the linear FVD model \eqref{eq:FVD} and the nonlinear ATG model \eqref{eq:FVD} with additive heterogeneity \eqref{eq:HCF2}. 
Indeed, theoretical results show that linear models with additive heterogeneity are systematically stable if the initial model is stable, see Proposition~\ref{prop:SC_HCFlin}. 
However, nonlinear models such as the ATG model may be unstable for large additive biases, especially negative ones, see \eqref{eq:SCatg2}.

\subsection{Numerical schemes}

The car-following models are simulated using an implicit Euler scheme for the positions of the vehicles, and an explicit Euler scheme for the speeds. 
In both schemes the time step is $\delta t=0.01$~s. 

\paragraph{Full Velocity Difference Model}
For the FVD model \eqref{eq:FVD} with additive heterogeneity \eqref{eq:HCF1}, the scheme for the $n$-th vehicle, $n\in\{1,\dots,N\}$, is given by:
\begin{equation}\label{eq:FVDnum}
    \begin{cases}
    ~x_n(t+\delta t)= x_n(t)+\delta t v_n(t+\delta t),\\
    ~\displaystyle v_n(t+\delta t) = v_n(t)+\delta t\Bigl[\lambda_1 \Bigl(\frac{g_n(t)}{T}-v_n(t)\Bigr) - \lambda_2\Delta v_n(t)+b_n\Bigr],
    \end{cases}
\end{equation}
where $T=1$~s, $\lambda_1=1$~s$^{-1}$, 
$\lambda_2=0.5$~s$^{-1}$. 
Such a setting is such that the stability condition of the homogeneous FVD model \eqref{eq:SC-FVD} holds and is critical.
The additive biases $(b_1,\dots,b_N)$ are assumed to be initially randomly distributed on $[-5,5]$~m/s, independent, and constant over the time. 
The FVD model, being linear, remains stable even in the presence of large biases (see Proposition~\ref{prop:SC_HCFlin}).

\paragraph{Adaptive Time Gap Model}

The ATG model \eqref{eq:ATG} is nonlinear and has a singularity when the distance gap $g$ is zero. 
However, negative gaps (i.e., collisions between the vehicles) can be expected during the simulations, especially if the model is unstable.
To extend the definition of the model even in the cases where the gap is negative, we reformulate the ATG model as
\begin{align*}
\dot{v}_n(t) &= \lambda v_n(t) \Bigl(1-\frac{Tv_n(t)}{g_n(t)}\Bigr) - \frac{v_n(t)\Delta v_n(t)}{g_n(t)},\\
    &=\frac{1}{T_n(t)} \bigl(\lambda (g_n(t)-Tv_n(t))-\Delta v_n(t)\bigr),
\end{align*}
where 
\begin{equation*}
    T_n(t)=\frac{g_n(t)}{v_n(t)},\qquad v_n(t)>0,
\end{equation*}
is the time gap of the $n$-th vehicle, $n\in\{1,\dots,N\}$.
We bound the time gap in $[\tmin,\tmax]$ using the mollifier $T_\varepsilon\big(g_n(t),v_n(t)\big)$ given by 
\begin{equation}
        T_\varepsilon(g,v)= f_\varepsilon \biggl(\tmin,f_{-\varepsilon}\Bigl(\tmax,\frac{g}{f_{\varepsilon}(0,v)}\Bigr)\biggr),
\end{equation}
where $f_\varepsilon$ is the $\operatorname{LogSumExp}$ function 
\begin{equation}\label{eq:LogSumExp}
   f_\varepsilon(a,b)=\varepsilon \log\bigl(\exp(a/\varepsilon)+\exp(b/\varepsilon)\bigr),\qquad \varepsilon=0.01.
\end{equation}
The function $f_\varepsilon(a,b)$ converges to the maximum of $a$ and $b$ as $\varepsilon\to0^+$, and to the minimum as $\varepsilon\to0^-$. 
This smoothing of the dynamics avoids singularities when the vehicles collide or when the speed is negative.
The numerical scheme for the extended ATG model with additive heterogeneity for the $n$-th vehicle, $n\in\{1,\dots,N\}$, is finally
\begin{equation}\label{eq:ATGnum}
    \begin{cases}
    ~x_n(t+\delta t) = x_n(t)+\delta t v_n (t+\delta t),\\
    ~\displaystyle v_n(t+\delta t) = v_n(t) + \delta t
    \Bigl[\frac{\lambda \bigl(g_n(t)-Tv_n(t)\bigr)-\Delta v_n(t)}{T_\varepsilon\bigl(g_n(t),v_n(t)\bigr)}+b_n\Bigr].
    \end{cases}
\end{equation}

The parameter values for the simulation are $\lambda=0.2$~s$^{-1}$, $\tmin=0.1$~s, and $\tmax=4$~s, while the desired time gap is $T=1$~s, as for the FVD model. 
The additive biases $(b_1,\dots,b_N)$ are assumed to be initially randomly distributed on $[-1,1]$~m/s, independent, and constant over the time. 
Such parameterization is such that the stability condition \eqref{eq:SCatg2} does not hold, i.e., the ATG model \eqref{eq:ATGnum} with additive heterogeneity is unstable.

\subsection{Simulation results}

The simulations are performed with 20 vehicles of length $\ell=5$~m on a segment of length $L=230$~m with periodic boundaries, starting from a uniform initial condition. Such a setting is similar to the scenarios of the experiments carried out in Japan in 2007 \cite{sugiyama2008traffic} and more recently in the United States \cite{stern_DissipationStopandgoWaves_2018}, which show rapid formation of stop-and-go waves.

We compare the dynamics obtained with the linear FVD model, which remains stable even in the presence of acceleration biases, with the nonlinear ATG model, where the biases induce an instability. 
The trajectories of the vehicles over the first 200 seconds are shown in Fig.~\ref{fig:Traj}, while the standard deviations of the vehicle distances and speeds are shown in Fig.~\ref{fig:SD}. 

On the one hand, the system converges to an equilibrium with heterogeneous gaps, due to acceleration biases $(b_1,\dots,b_N)$, see Fig.~\ref{fig:Traj}, left panel. 
The gap deviation converges to a constant while the speed deviation converges to zero (see Fig.~\ref{fig:SD}, grey curves). 
The biases affect the equilibrium solution. 
However, the FVD model, being linear, remains stable (see Proposition~\ref{prop:SC_HCFlin}).

On the other hand, the biases affect both the equilibrium and the stability of the nonlinear ATG model. 
Indeed, the stability condition \eqref{eq:SCatg2} no longer holds for the chosen parameter setting.
The dynamics show the emergence of a stop-and-go wave (see Fig.~\ref{fig:Traj}, right panel), while the gap and speed standard deviations converge to a limit cycle (see Fig.~\ref{fig:Traj}, blue curves).

\begin{figure}[!ht]
    \centering
    \includegraphics{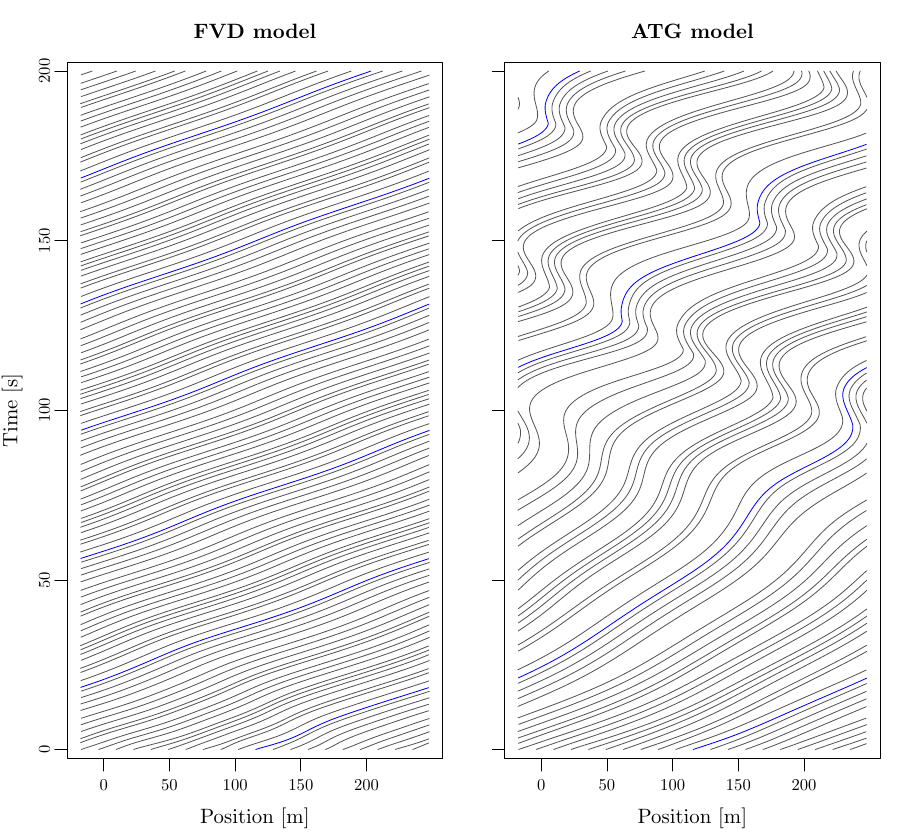}
    \caption{Trajectories of 20 vehicles on a segment of length $L=230$~m with periodic boundaries from uniform initial conditions. 
    Left panel:
    Linear FVD model \eqref{eq:FVDnum} with additive heterogeneity (stable). 
    Right panel: 
    nonlinear ATG  model \eqref{eq:ATGnum} with additive heterogeneity (unstable). 
    The FVD model remains stable even even in the presence of acceleration biases and converges to an equilibrium with heterogeneous gaps, while the biases perturb the stability of the ATG model, which converges to a limit cycle with a stop-and-go wave propagating downstream.}
    \label{fig:Traj}
\end{figure}

\begin{figure}[!ht]
    \centering\bigskip
    \includegraphics{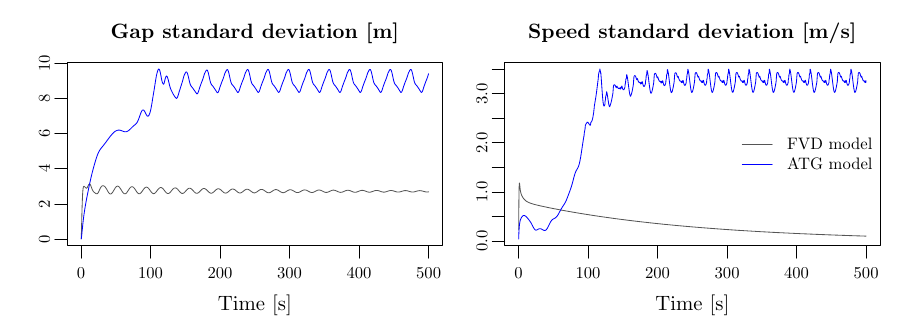}
    \caption{Time evolution of the empirical speed (left panel) and gap (right panel) standard deviations for the two scenarios shown in Fig.~\protect\ref{fig:Traj}. The gap standard deviation converges and the speed standard deviation tends to zero for the stable FVD model (grey curves), while the system converges to a limit cycle with a stop-and-go wave for the unstable ATG model (blue curves).}
    \label{fig:SD}
\end{figure}

\section{Conclusion}  \label{sec:Ccl}
The results presented in this paper show that simple static heterogeneity mechanisms in driver behaviour (quenched disorder), here an individual additive bias in acceleration, see \eqref{eq:HCF1}, can affect the stability of nonlinear models and lead to the emergence of stop-and-go waves. 
Such a result is specific to nonlinear models, as the stability of linear models is inherently robust to additive acceleration biases. 
However, scalar heterogeneity models \eqref{eq:HCF2} can affect the stability of both linear and nonlinear models, see \eqref{eq:SC1}. 
In general, we can conclude that static heterogeneity in driver behaviour can lead to traffic instability and stop-and-go dynamics.
However, other factors can also affect the stability of a line of vehicles, such as control delay, stochastic noise or, again, limited acceleration capacity. 
Stable car-following models and adaptive cruise controllers should consider these factors in a unified framework.


\section*{Declarations}



\subsection*{CRediT author statement}

\textit{M. Ehrhardt:} 
Validation, 
Formal analysis, 
Investigation, 
Writing -- Original Draft, 
Writing -- Review \& Editing, 
Project administration.
\textit{A. Tordeux:}
Conceptualisation, 
Methodology, 
Validation, 
Formal analysis, 
Investigation,  
Writing -- Original Draft, 
Writing -- Review \& Editing,  
Software.





\subsection*{Data availability} 

All information analyzed or generated, which would support the results of this work are available in this article.
No data was used for the research described in the article.

\subsection*{Conflict of interest} 
The authors declare that there are no problems or conflicts 
of interest between them that may affect the study in this paper.

\subsection*{Acknowledgements} 
The authors thank Oscar Dufour, Alexandre Nicolas, and David Rodney for fruitful discussions that initiated the work presented in this paper.


\bibliographystyle{acm}
\bibliography{refs}

\begin{thebibliography}{10}

\bibitem{ackermann_tordeux24}
{\sc Ackermann, J., Ehrhardt, M., Kruse, T., and Tordeux, A.}
\newblock Stabilisation of stochastic single-file dynamics using port-{H}amiltonian systems, 2024.
\newblock Accepted for publication in IFAC-PapersOnLine (26th International Symposium on Mathematical Theory of Networks and Systems).

\bibitem{aguilera2014new}
{\sc Aguil{\'e}ra, V., and Tordeux, A.}
\newblock A new kind of fundamental diagram with an application to road traffic emission modeling.
\newblock {\em J. Adv. Transport. 48}, 2 (2014), 165--184.

\bibitem{andre2000driving}
{\sc Andr{\'e}, M., and Hammarstr{\"o}m, U.}
\newblock Driving speeds in {Europe} for pollutant emissions estimation.
\newblock {\em Transport. Res. Part D: Transp. Environ. 5}, 5 (2000), 321--335.

\bibitem{bando1998analysis}
{\sc Bando, M., Hasebe, K., Nakanishi, K., and Nakayama, A.}
\newblock Analysis of optimal velocity model with explicit delay.
\newblock {\em Phys. Rev. E 58}, 5 (1998), 5429.

\bibitem{bando_DynamicalModelTraffic_1995}
{\sc Bando, M., Hasebe, K., Nakayama, A., Shibata, A., and Sugiyama, Y.}
\newblock Dynamical model of traffic congestion and numerical simulation.
\newblock {\em Phys. Rev. E 51}, 2 (1995), 1035--1042.

\bibitem{BarlovicSSS98}
{\sc Barlovic, R., Santen, L., Schadschneider, A., and Schreckenberg, M.}
\newblock Metastable states in cellular automata for traffic flow.
\newblock {\em Euro. Phys. J. B 5\/} (1998), 793--800.

\bibitem{chandler1958traffic}
{\sc Chandler, R., Herman, R., and Montroll, E.}
\newblock Traffic dynamics: studies in car following.
\newblock {\em Oper. Res. 6}, 2 (1958), 165--184.

\bibitem{cordes_ModelingStopandGoWaves_2020}
{\sc Cordes, J.}
\newblock Modeling of {{Stop-and-Go Waves}} in {{Pedestrian Dynamics}}, 2020.

\bibitem{cordes2023single}
{\sc Cordes, J., Chraibi, M., Tordeux, A., and Schadschneider, A.}
\newblock Single-file pedestrian dynamics: a review of agent-following models.
\newblock {\em Crowd Dynamics, Volume 4: Analytics and Human Factors in Crowd Modeling\/} (2023), 143--178.

\bibitem{Ehrhardt2024}
{\sc Ehrhardt, M., Kruse, T., and Tordeux, A.}
\newblock The collective dynamics of a stochastic port-{H}amiltonian self-driven agent model in one dimension.
\newblock {\em ESAIM: Math. Model. Numer. Anal. 58\/} (2024), 515--544.

\bibitem{friesen2021spontaneous}
{\sc Friesen, M., Gottschalk, H., R{\"u}diger, B., and Tordeux, A.}
\newblock Spontaneous wave formation in stochastic self-driven particle systems.
\newblock {\em SIAM J. Appl. Math. 81}, 3 (2021), 853--870.

\bibitem{gazis1961nonlinear}
{\sc Gazis, D., Herman, R., and Rothery, R.}
\newblock Nonlinear follow-the-leader models of traffic flow.
\newblock {\em Oper. Res. 9}, 4 (1961), 545--567.

\bibitem{gunter2020commercially}
{\sc Gunter, G., Gloudemans, D., Stern, R., et~al.}
\newblock Are commercially implemented adaptive cruise control systems string stable?
\newblock {\em IEEE Trans. Intell. Transport. Systems 22}, 11 (2020), 6992--7003.

\bibitem{herman_TrafficDynamicsAnalysis_1959}
{\sc Herman, R., Montroll, E., Potts, R.~B., and Rothery, R.}
\newblock Traffic dynamics: {Analysis} of stability in car following.
\newblock {\em Oper. Res. 7}, 1 (1959), 86--106.

\bibitem{huang2018instability}
{\sc Huang, Y.-X., Guo, N., Jiang, R., and Hu, M.-B.}
\newblock Instability in car-following behavior: new {N}agel-{S}chreckenberg type cellular automata model.
\newblock {\em J. Stat. Mech.: Theory Exper. 2018}, 8 (2018), 083401.

\bibitem{jiang2021dampen}
{\sc Jiang, L., Xie, Y., Wen, X., Chen, D., Li, T., and Evans, N.}
\newblock Dampen the stop-and-go traffic with connected and automated vehicles -- a deep reinforcement learning approach.
\newblock In {\em 2021 7th International Conference on Models and Technologies for Intelligent Transportation Systems (MT-ITS)\/} (2021), IEEE, pp.~1--6.

\bibitem{jiang_FullVelocityDifference_2001}
{\sc Jiang, R., Wu, Q., and Zhu, Z.}
\newblock Full velocity difference model for a car-following theory.
\newblock {\em Phys. Rev. E 64\/} (2001), 017101.

\bibitem{kaupuvzs2005zero}
{\sc Kaupu{\v{z}}s, J., Mahnke, R., and Harris, R.}
\newblock Zero-range model of traffic flow.
\newblock {\em Phys. Rev. E 72}, 5 (2005), 056125.

\bibitem{khelfa2022heterogeneity}
{\sc Khelfa, B., Korbmacher, R., Schadschneider, A., and Tordeux, A.}
\newblock Heterogeneity-induced lane and band formation in self-driven particle systems.
\newblock {\em Scientific reports 12}, 1 (2022), 4768.

\bibitem{khound2021extending}
{\sc Khound, P., Will, P., Tordeux, A., and Gronwald, F.}
\newblock Extending the adaptive time gap car-following model to enhance local and string stability for adaptive cruise control systems.
\newblock {\em J. Intell. Transport. Syst.\/} (2021), 1--21.

\bibitem{kometani1958stability}
{\sc Kometani, E., and Sasaki, T.}
\newblock On the stability of traffic flow ({Report-I}).
\newblock {\em J. Oper. Res. Soc. Japan 2}, 1 (1958), 11--26.

\bibitem{kreidieh2018dissipating}
{\sc Kreidieh, A., Wu, C., and Bayen, A.}
\newblock Dissipating stop-and-go waves in closed and open networks via deep reinforcement learning.
\newblock In {\em 2018 21st International Conference on Intelligent Transportation Systems (ITSC)\/} (2018), IEEE, pp.~1475--1480.

\bibitem{krusemann2014first}
{\sc Kr{\"u}semann, H., Godec, A., and Metzler, R.}
\newblock First-passage statistics for aging diffusion in systems with annealed and quenched disorder.
\newblock {\em Phys. Rev. E 89}, 4 (2014), 040101.

\bibitem{Ke-Ping_2004}
{\sc Li, K.-P., and Gao, Z.-Y.}
\newblock Noise-induced phase transition in traffic flow*.
\newblock {\em Commun. Theor. Phys. 42}, 3 (sep 2004), 369.

\bibitem{makridis2021openacc}
{\sc Makridis, M., Mattas, K., Anesiadou, A., and Ciuffo, B.}
\newblock {OpenACC}. an open database of car-following experiments to study the properties of commercial {ACC} systems.
\newblock {\em Transport. Res. Part C: Emerg. Techn. 125\/} (2021), 103047.

\bibitem{nagatani1998delay}
{\sc Nagatani, T., and Nakanishi, K.}
\newblock Delay effect on phase transitions in traffic dynamics.
\newblock {\em Phys. Rev. E 57}, 6 (1998), 6415.

\bibitem{nakayama2009metastability}
{\sc Nakayama, A., Fukui, M., Kikuchi, M., et~al.}
\newblock Metastability in the formation of an experimental traffic jam.
\newblock {\em New J. Phys. 11}, 8 (2009), 083025.

\bibitem{newell1961nonlinear}
{\sc Newell, G.}
\newblock Nonlinear effects in the dynamics of car following.
\newblock {\em Oper. Res. 9}, 2 (1961), 209--229.

\bibitem{ngoduy2015effect}
{\sc Ngoduy, D.}
\newblock Effect of the car-following combinations on the instability of heterogeneous traffic flow.
\newblock {\em Transportmetrica B: Transport Dyn. 3}, 1 (2015), 44--58.

\bibitem{ngoduy2019langevin}
{\sc Ngoduy, D., Lee, S., Treiber, M., Keyvan-Ekbatani, M., and Vu, H.}
\newblock Langevin method for a continuous stochastic car-following model and its stability conditions.
\newblock {\em Transport. Res. Part C: Emerg. Techn. 105\/} (2019), 599--610.

\bibitem{orosz2004global}
{\sc Orosz, G., Wilson, R., and Krauskopf, B.}
\newblock Global bifurcation investigation of an optimal velocity traffic model with driver reaction time.
\newblock {\em Phys. Rev. E 70}, 2 (2004), 026207.

\bibitem{orosz2010traffic}
{\sc Orosz, G., Wilson, R., and St{\'e}p{\'a}n, G.}
\newblock Traffic jams: dynamics and control.
\newblock {\em Phil. Trans. Royal Soc. A: Math., Phys. Engrg. Sci. 368}, 1928 (2010), 4455--4479.

\bibitem{pipes_OperationalAnalysisTraffic_1953}
{\sc Pipes, L.}
\newblock An operational analysis of traffic dynamics.
\newblock {\em J. Appl. Phys. 24}, 3 (1953), 274--281.

\bibitem{Reuschel1950fahrzeugbewegungen}
{\sc Reuschel, A.}
\newblock Fahrzeugbewegungen in der {K}olonne.
\newblock {\em \"Osterreichisches Ingenieur Archiv 4\/} (1950), 193--215.

\bibitem{schadschneider_StochasticTransportComplex_2010}
{\sc Schadschneider, A., Chowdhury, D., and Nishinari, K.}
\newblock {\em Stochastic {{Transport}} in {{Complex Systems}}. {{From Molecules}} to {{Vehicles}}}.
\newblock Elsevier, 2010.

\bibitem{stern2019quantifying}
{\sc Stern, R., Chen, Y., Churchill, M., et~al.}
\newblock Quantifying air quality benefits resulting from few autonomous vehicles stabilizing traffic.
\newblock {\em Transport. Res. Part D: Transp. Environm. 67\/} (2019), 351--365.

\bibitem{stern_DissipationStopandgoWaves_2018}
{\sc Stern, R.~E., Cui, S., Delle~Monache, M., et~al.}
\newblock Dissipation of stop-and-go waves via control of autonomous vehicles: {{Field}} experiments.
\newblock {\em Transport. Res. Part C: Emerg. Techn. 89\/} (2018), 205--221.

\bibitem{sugiyama2008traffic}
{\sc Sugiyama, Y., Fukui, M., Kikuchi, M., et~al.}
\newblock Traffic jams without bottlenecks-experimental evidence for the physical mechanism of the formation of a jam.
\newblock {\em New J. Phys. 10}, 3 (2008), 033001.

\bibitem{tadaki2013phase}
{\sc Tadaki, S.-I., Kikuchi, M., Fukui, M., et~al.}
\newblock Phase transition in traffic jam experiment on a circuit.
\newblock {\em New J. Phys. 15}, 10 (2013), 103034.

\bibitem{tateishi2020quenched}
{\sc Tateishi, A., Ribeiro, H., Sandev, T., Petreska, I., and Lenzi, E.}
\newblock Quenched and annealed disorder mechanisms in comb models with fractional operators.
\newblock {\em Phys. Rev. E 101}, 2 (2020), 022135.

\bibitem{tomer2000presence}
{\sc Tomer, E., Safonov, L., and Havlin, S.}
\newblock Presence of many stable nonhomogeneous states in an inertial car-following model.
\newblock {\em Phys. Rev. Lett. 84}, 2 (2000), 382.

\bibitem{tordeux2018traffic}
{\sc Tordeux, A., Costeseque, G., Herty, M., and Seyfried, A.}
\newblock From traffic and pedestrian follow-the-leader models with reaction time to first order convection-diffusion flow models.
\newblock {\em SIAM J. Appl. Math. 78}, 1 (2018), 63--79.

\bibitem{tordeux2010adaptive}
{\sc Tordeux, A., Lassarre, S., and Roussignol, M.}
\newblock An adaptive time gap car-following model.
\newblock {\em Transport. Res. Part B: Method. 44}, 8-9 (2010), 1115--1131.

\bibitem{tordeux_LinearStabilityAnalysis_2012}
{\sc Tordeux, A., Roussignol, M., and Lassarre, S.}
\newblock Linear stability analysis of first-order delayed car-following models on a ring.
\newblock {\em Phys. Rev. E 86}, 3 (2012), 036207.

\bibitem{Tordeux2016b}
{\sc Tordeux, A., and Schadschneider, A.}
\newblock White and relaxed noises in optimal velocity models for pedestrian flow with stop-and-go waves.
\newblock {\em J. Phys. A: Math. Theor. 18\/} (2016), 185101.

\bibitem{treiber2009hamilton}
{\sc Treiber, M., and Helbing, D.}
\newblock Hamilton-like statistics in onedimensional driven dissipative many-particle systems.
\newblock {\em Euro. Phys. J. B 68\/} (2009), 607--618.

\bibitem{treiber_CongestedTrafficStates_2000}
{\sc Treiber, M., Hennecke, A., and Helbing, D.}
\newblock Congested traffic states in empirical observations and microscopic simulations.
\newblock {\em Phys. Rev. E 62}, 2 (2000), 1805--1824.

\bibitem{treiber2013traffic}
{\sc Treiber, M., and Kesting, A.}
\newblock Traffic flow dynamics.
\newblock {\em Traffic Flow Dyn.: Data, Mod. Simul.\/} (2013), 983--1000.

\bibitem{treiber2017intelligent}
{\sc Treiber, M., and Kesting, A.}
\newblock The intelligent driver model with stochasticity-new insights into traffic flow oscillations.
\newblock {\em Transport. Res. Proc. 23\/} (2017), 174--187.

\bibitem{wagner2011time}
{\sc Wagner, P.}
\newblock A time-discrete harmonic oscillator model of human car-following.
\newblock {\em Euro. Phys. J. B 84\/} (2011), 713--718.

\bibitem{wang2020stability}
{\sc Wang, Y., Li, X., Tian, J., and Jiang, R.}
\newblock Stability analysis of stochastic linear car-following models.
\newblock {\em Transport. Sci. 54}, 1 (2020), 274--297.

\bibitem{wilson_CarfollowingModelsFifty_2011}
{\sc Wilson, R., and Ward, J.}
\newblock Car-following models: Fifty years of linear stability analysis -- a mathematical perspective.
\newblock {\em Transport. Plan. Techn. 34}, 1 (2011), 3--18.

\end{thebibliography}

\end{document}